\documentclass[prx,twocolumn,secnumarabic,balancelastpage,amsmath,amssymb,longbibliography]{revtex4-2}

\usepackage{graphicx}
\usepackage[dvipsnames]{xcolor}
\usepackage[colorlinks,linkcolor=black,citecolor=black,urlcolor=black,filecolor=black]{hyperref}
\usepackage{amsmath, amssymb, amsthm}
\usepackage[normalem]{ulem}
\usepackage{wasysym}
\usepackage{lipsum}
\usepackage{outlines}
\usepackage{enumerate}
\usepackage{setspace}
\usepackage{tabularx} 
\usepackage{verbatim}
\usepackage{xcolor, soul}
\usepackage{physics}
\usepackage[export]{adjustbox}
\usepackage{tikz}
\usetikzlibrary{decorations.pathreplacing}
\usetikzlibrary{arrows,arrows.meta,calc,shapes.geometric,shapes.misc}
\tikzset{
	>=stealth',
	help lines/.style={dashed, thick},
	important line/.style={thick},
	connection/.style={thick, dotted},
}
\DeclareMathAlphabet{\mymathbb}{U}{BOONDOX-ds}{m}{n}
\usepackage{dsfont}
\usepackage{braket}
\hypersetup{
    pdfstartview={FitH}, 
    colorlinks=true, 
    linkcolor=blue, 
    citecolor=blue, 
    filecolor=magenta,
    urlcolor=blue
}
\usepackage{algorithm}
\usepackage{algpseudocode}

\usepackage{wrapfig}
\newtheorem{thm}{Theorem}
\newtheorem{theorem}[thm]{Theorem}
\newtheorem{lemma}[thm]{Lemma}
\newtheorem{definition}[thm]{Definition}

\newcommand{\lop}[1]{\overline{#1}}
\newcommand{\lket}[1]{\ket{\overline{#1}}}
\newcommand{\lcnot}{\overline{\text{CNOT}}}

\newcommand{\pth}{p_\text{th}}
\newcommand{\pthclifforddeeper}{0.690(2)\%}
\newcommand{\pthclifforddeep}{0.718(2)\%}
\newcommand{\pthmemory}{0.808(1)\%}
\newcommand{\pthdistillation}{0.817(4)\%}

\newcommand{\nmeas}{m}

\newcommand{\plogical}{P_{L}}

\begin{document}

\title{Fast correlated decoding of transversal logical algorithms}

\author{
Madelyn~Cain$^{1,*,\dagger}$, 
Dolev~Bluvstein$^{1,*}$, 
Chen Zhao$^{2,*}$, 
Shouzhen~Gu$^{3, 4}$, 
Nishad~Maskara$^{1}$, 
Marcin~Kalinowski$^{1}$, 
Alexandra~A.~Geim$^{1}$, 
Aleksander~Kubica$^{3,4}$, 
Mikhail~D.~Lukin$^{1,\ddagger}$, 
and Hengyun~Zhou$^{2,\S}$
}

\affiliation{
$^1$Department of Physics, Harvard University, Cambridge, Massachusetts 02138, USA\\
$^2$QuEra Computing Inc., 1284 Soldiers Field Road, Boston, MA, 02135, US\\
$^3$Department of Applied Physics, Yale University, New Haven, Connecticut 06511, USA\\
$^4$Yale Quantum Institute, Yale University, New Haven, Connecticut 06511, USA\\
$^*$These authors contributed equally, $^\dagger$maddiecain10@gmail.com,
$^\ddagger$lukin@physics.harvard.edu,
$^\S$hyzhou@quera.com
}

\date{\today}

\begin{abstract}
Quantum error correction (QEC) is required for large-scale computation, but incurs a significant resource overhead. 
Recent advances have shown that by jointly decoding logical qubits in algorithms composed of transversal gates, the number of syndrome extraction rounds can be reduced by a factor of the code distance $d$, at the cost of increased classical decoding complexity.
Here, we reformulate the problem of decoding transversal circuits by directly decoding relevant logical operator products as they propagate through the circuit.
This procedure transforms the decoding task into one closely resembling that of a single-qubit memory propagating through time. 
The resulting approach leads to fast decoding and reduced problem size while maintaining high performance. 
Focusing on the surface code, we prove that this method enables fault-tolerant decoding with minimum-weight perfect matching, and benchmark its performance on example circuits including magic state distillation.
We find that the threshold is comparable to that of a single-qubit memory, and that the total decoding run time can be, in fact, less than that of conventional lattice surgery. 
Our approach enables fast correlated decoding, providing a pathway to directly extend single-qubit QEC techniques to transversal algorithms.
\end{abstract}

\maketitle
\section{Introduction}
Quantum error correction (QEC) is believed to be essential for large-scale quantum computation~\cite{shor1994algorithms, preskill1998reliable, dennis2002topological, nielsen2010quantum, kitaev2003fault}.
Recent experiments have realized QEC across several different systems, demonstrating the hallmark exponential error suppression with increasing code distance $d$ below threshold physical error rates, and enabling higher-fidelity execution of tailored algorithms~\cite{bluvstein2024logical,acharya2024quantum,reichardt2024logical,ryan-anderson2024high,ryan-anderson2022implementing}.
These advances mark a turning point, as efficient QEC techniques become paramount to implementing error-corrected quantum algorithms in practical systems.
Moreover, these experiments highlight the central role of the QEC \textit{decoder}, a classical algorithm that uses measurement results to infer and correct errors. 
The decoder has a significant impact on the practical performance of computation with logical qubits: its accuracy affects whether a system is below the threshold, while its speed directly enters into the execution speed of the computation. 

In particular, a key tool leveraged in these experiments is transversal gates, which implement a logical operation by applying tensor products of single- or two-qubit physical gates~\cite{bluvstein2022quantum, postler2022demonstration, ryan-anderson2022implementing, honciuc2024implementing, wang2023fault, bluvstein2024logical}. 
Recent advances have shown that the number of syndrome extraction (SE) rounds required per transversal gate can be significantly reduced from $O(d)$ to $O(1)$ using correlated decoding, in which multiple logical qubits are decoded jointly~\cite{bluvstein2024logical, cain2024correlated, zhou2024algorithmic}.
However, existing strategies for correlated decoding face key limitations: either the performance is reduced, or the run time rapidly increases with system size. 
These trade-offs stem from two core challenges.
First, \textit{hyperedges} in the decoding graph arising from syndrome measurement errors between transversal operations~\cite{cain2024correlated,sahay2024error,wan2024iterative,wu2024hypergraph} prevent the use of fast and well-understood decoders based on minimum-weight perfect matching (MWPM)~\cite{dennis2002topological,higgott2025sparse,wu2023fusion}.
Second, existing approaches assume jointly decoding the entire circuit in order to ensure fault tolerance, resulting in a rapidly increasing decoding volume. 
These difficulties raise fundamental questions about the trade-offs between quantum resources and classical decoding complexity, and pose significant practical barriers to decoding large-scale algorithms with transversal gates.

In this Article, we present a strategy to decode individual \textit{logical operator products}, rather than the circuit as a whole.
By isolating to only a subset of syndrome measurements, the hyperedges from measurement errors are directly reduced to simple edges, such that the decoding problem closely resembles that of a single logical qubit propagating through time.
Specializing to the surface code, this enables the entire circuit to be decoded using MWPM.
In many practical settings, the decoding volume is substantially reduced, and the corrections can be committed in software such that stabilizers do not need to be re-decoded.
Furthermore, because the number of stabilizer measurements is reduced by a factor of $d$, we find that the total decoding runtime can be, in fact, less than that of conventional lattice surgery.
We prove the fault tolerance of our decoding strategy, and benchmark its performance on example circuits, including random Clifford circuits and magic state distillation.
Our results demonstrate thresholds and decoding run times comparable to those of a single-qubit memory and modular lattice surgery decoding, respectively~\cite{bombin2023modular,tan2023scalable,skoric2023parallel}. 
These findings establish procedures for fast and accurate decoding of transversal algorithms, and provide a framework for extending QEC techniques from the memory setting to large-scale logical algorithms.

\section{Decoding strategy}

\subsection{Decoding reliable logical products \label{subsec:decoding_strategy}}
Universal quantum computation can be performed via an adaptive transversal Clifford circuit acting on logical Pauli states ($\lket{0}$ or $\lket{+}$) and magic states \mbox{$\lket{T} =(\lket{0} + e^{i\pi/4}\lket{1})/\sqrt{2}$} with $\lop{Z}$ and $\lop{X}$ basis measurements [Fig.~\ref{fig:figure_1}(a)], where the overline indicates logical qubits.
Interestingly, such universal circuits can be fault-tolerantly implemented with only $O(1)$ SE rounds per transversal logical operation and logical Pauli state initialization, assuming the $\lket{T}$ states are prepared fault-tolerantly~\cite{zhou2024algorithmic}.
The essence of this can be understood by tracking how logical Pauli operators propagate through the circuit. 
Although the circuit generally involves mid-circuit measurements and conditional gates, upon execution it is a transversal Clifford circuit (acting on Pauli and non-Pauli inputs), so such operators can be tracked deterministically.
Each measurement, or more generally \textit{products} of measurements, has an associated logical Pauli operator which can be back-propagated through the circuit.
Back-propagated measurements which anti-commute with a logical Pauli initialization are 50/50 random.
As such, these measurements do not provide information about the quantum state; individually, they can be assigned uniformly at random $\pm 1$ outcomes and thus do not need to be decoded.
All other measurement products, which terminate on $\lket{T}$ states or logical Pauli states in the same basis, contain non-trivial information.
A basis of these non-trivial measurements must therefore be decoded reliably in order to reproduce the logical measurement distribution of the ideal circuit.

\begin{figure}[t!]
    \centering
    \includegraphics{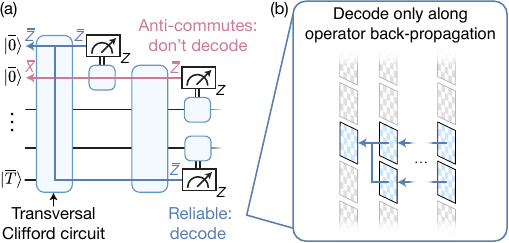}
    \caption{Decoding transversal algorithms.
    (a)~Logical measurements in universal computation can be correctly predicted by decoding \textit{reliable logical Pauli products} which, when back-propagated through the circuit, terminate at $\lket{T}$ states or logical Pauli states in the same basis.
    All other measurements can be assigned uniformly at random $\pm 1$ outcomes, as they terminate at a logical Pauli initialization in the anti-commuting basis.
    (b)~The reliable logical Pauli products can be decoded by tracing back their evolution through the circuit and decoding only the stabilizers along the resulting propagation path. 
    }
    \label{fig:figure_1}
\end{figure}

We refer to these measurement products of interest as \textit{reliable logical Pauli products}, because the same condition that makes the logical state well-defined also ensures their stabilizers at initialization are deterministically $+1$ (Sec.~\ref{subsec:fault_tolerance}).
Crucially, we find that for any Calderbank-Shor-Steane (CSS) code, these reliable Pauli products can be determined using only the stabilizer measurements along the back-propagation path of the operator [Fig.~\ref{fig:figure_1}(b)]. 
Restricting to these relevant stabilizers greatly simplifies the decoding problem, reducing its size while maintaining high performance.

We now outline our decoding strategy.
For each new logical measurement, the algorithm takes in a description of the circuit and the current measurement snapshots, and outputs an assignment of the new logical measurement using the following steps.
\begin{enumerate}
    \item Find a basis of logical measurement products over the  measurements so far, where each element is either a single measurement that is independently 50/50 random, or a reliable Pauli product (Sec.~\ref{sec:proof}).  
    If the current measurement is not part of a reliable Pauli product, assign a $\pm1$ value uniformly at random, and skip the remaining steps. 
    \item  Choose a reliable logical Pauli product that includes the current measurement and potentially already-assigned past measurements, and back-propagate it through the circuit.
    Record all stabilizer measurements in the same basis along the propagation path.
    \item Decode the reliable logical Pauli product using only these relevant stabilizer measurements.
    In the case of the surface code, MWPM can be used (Sec.~\ref{subsec:matchable}). 
    \item (Optional) In certain cases (see Sec.~\ref{subsec:volume}), the corrections can be committed in software, reducing the size of future decoding problems.
    \item Assign the current logical measurement result, which is equal to the product of the decoded reliable logical Pauli product and any already-assigned measurements in the product (including previous randomly-assigned measurements). 
\end{enumerate}
We give an explicit example of this strategy in Appendix~\ref{appendix:decoding_example}.
By calling this procedure each time a new logical qubit is measured, the logical bit string for the entire circuit is sampled from the ideal distribution.
Interestingly, running the procedure again on the same measured data can produce different values for individual logical measurement due to the classically-generated 50/50 randomness.
Nevertheless, the meaningful information, i.e.,~the values of the reliable logical Pauli products, are always the same.

In the following subsections, we will explain these steps in detail, including how this approach removes hyperedges from the decoding problem (Sec.~\ref{subsec:matchable}), enables reliable interpretation of logical measurement results (Sec.~\ref{subsec:fault_tolerance}), and can reduce the decoding volume (Sec.~\ref{subsec:volume}). 
We benchmark the performance of our approach numerically in Sec.~\ref{sec:numerics} and prove its fault tolerance in Sec.~\ref{sec:proof}.
These results establish a theoretical foundation for fast and accurate correlated decoding.

\begin{figure}[t!]
    \centering
    \includegraphics{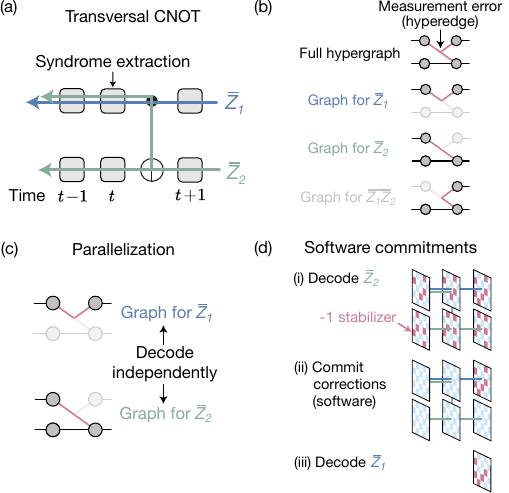}
    \caption{Matchable decoding and correction strategies.
    (a)~Logical Pauli products are decoded using stabilizer measurements along their back-propagation through the circuit. 
    (b)~During a transversal $\lcnot$, the checks (vertices) are $Z_1^{t-1}Z_1^{t}$ (top left), $Z_1^{t}Z_1^{t+1}$ (top right),  $Z_2^{t-1}Z_2^{t}$ (bottom left) and $Z_1^{t}Z_2^{t}Z_2^{t+1}$ (bottom right). 
    A measurement error on $Z_1^{t}$ (pink hyperedge) therefore flips three checks, complicating the decoding problem.
    By decoding only the checks along the back-propagation path of the logical operator, stabilizer measurement errors flip only two checks, enabling efficient decoding via MWPM.
    (c)~Logical Pauli products can be decoded independently in parallel.
    (d)~For some circuits (see Sec.~\ref{subsec:matchable}), the corrections for a logical Pauli product can be committed after decoding (i-ii), reducing the size of future decoding problems (iii). 
    }
    \label{fig:figure_2}
\end{figure}

\subsection{Constructing a matchable decoding problem \label{subsec:matchable}}
Here we describe how to remove hyperedges from the decoding problem.
We will decode a reliable Pauli product using only the stabilizer measurements in the same basis as the instantaneous logical operator along its backwards-propagated path through the circuit.
Only these stabilizer measurements are necessary: any error which can corrupt the logical product at its final measurement must also be detected by these stabilizers (Sec.~\ref{sec:proof}, Lemma~\ref{lemma:complete_subgraph}).

The decoder takes as input the stabilizer measurements and a \textit{decoding hypergraph}, which summarizes the circuit error model. 
The vertices of the hypergraph represent checks.
For simplicity, here we assume that at least one SE round occurs between transversal gates at each time step $t$, and discuss generalizations to multiple gates per round in Lemma~\ref{lemma:fewer_se} of Appendix~\ref{appendix:proof}.
Each check is then equal to the product of a stabilizer measurement and its backwards-propagated operator at the previous time step, if one exists. 
A check will flip from $+1$ to $-1$ if it detects an error.
Errors correspond to hyperedges connecting the checks they flip.

In practice, the weight of the hyperedges (how many vertices they flip) has significant implications on the complexity of the decoding problem. 
Circuits where all hyperedges have weight at most two (\textit{simple} edges) can be efficiently decoded using algorithms such as MWPM~\cite{dennis2002topological,higgott2025sparse,wu2023fusion} and union find~\cite{delfosse2017almost}. 
For example, these decoders can be applied to surface code computations with lattice surgery gates, where both data qubit and stabilizer measurement errors correspond to simple edges (any higher-weight hyperedges arising from, e.g.,~correlated errors can always be decomposed into these simple edges).
Higher-weight hyperedges for which this decomposition is not possible require more complex algorithms and can incur significantly longer run times in practice~\cite{panteleev2019degenerate,roffe2020decoding,delfosse2022toward,cain2024correlated,wu2024hypergraph}.

Stabilizer measurement errors are simple edges in a memory setting because a stabilizer measurement error will only flip two checks, comparing its measurement at time $t$ to times $t-1$ and $t+1$.
In contrast, transversal gates correlate stabilizer values across multiple logical qubits, such that a single stabilizer measurement error can lead to a weight-three hyperedge that is irreducible when decoding the entire logical circuit~\cite{cain2024correlated}.
To see this, consider a transversal $\lcnot$ with three rounds of surrounding $Z$ stabilizer measurements, as shown in Fig.~\ref{fig:figure_2}(a).
As illustrated in Fig.~\ref{fig:figure_2}(b) (top), the control qubit has check vertices $Z_1^{t-1}Z_1^{t}$ (top left) and $Z_1^{t}Z_1^{t+1}$ (top right), and the target qubit has checks $Z_2^{t-1}Z_2^{t}$ (bottom left) and $Z_1^{t}Z_2^{t}Z_2^{t+1}$ (bottom right) ($Z_i^{t}$ denotes a generic $Z$ stabilizer on logical qubit $i$ at time $t$).
Therefore, a measurement error on $Z_1^{t}$ will flip three checks (pink hyperedge). 

However, by decoding only the checks that include stabilizer measurements relevant to a logical Pauli product, such hyperedges are entirely avoided.
Crucially, only two of the three checks in the hyperedge are required for any choice of logical product, reducing the hyperedge weight to two (Fig.~\ref{fig:figure_2}(b), bottom).
This simplification arises from the fact that transversal Clifford gates transform the logical operator and stabilizers in the same way.
As a result, the checks can be chosen to track how the logical locally transforms between $t- 1$ and $t$, as well as $t$ and $t+1$, similar to the case of a single-qubit memory.
Because each stabilizer is only involved in two checks, if it is measured incorrectly only these two checks will flip, resulting in a simple edge.
In Appendix~\ref{appendix:decoding_hypergraph}, we show this pattern explicitly for the remaining transversal Clifford gates in the surface code, including the Hadamard $\lop{H}$ and phase $\lop{S}$ gates.

These observations apply broadly to transversal gates in CSS codes. 
Again, the decoding hypergraph will track the logical Pauli product through space and time. 
The space-like hyperedges, due to physical errors on the data qubits, will have the same structure as data qubit errors on a single copy of the original code without any logic gates. 
The time-like edges, corresponding to stabilizer measurement errors, will always be simple edges. 
Thus, we expect that in many cases, standard decoders used to decode a single logical qubit of a particular code can be promoted to decode transversal algorithms. 
For the two-dimensional surface and color codes, variants of MWPM can be applied~\cite{dennis2002topological,kubica2023efficient,delfosse2014decoding,lee2025color,gidney2023new}, enabling fast correlated decoding of the transversal logical algorithm. 

\subsection{Fault tolerance of the decoding strategy \label{subsec:fault_tolerance}}
The preceding discussion shows that the decoder can operate with the fast speed of MWPM.
Here, we explain why it is also fault-tolerant and achieves distance $O(d)$, akin to the results in Ref.~\cite{zhou2024algorithmic}.
We prove these findings in Sec.~\ref{sec:proof}.

One might initially be concerned that only $O(1)$ SE rounds between initialization and measurement, as opposed to $O(d)$, may not be sufficient maintain fault tolerance against stabilizer measurement errors. 
This originates from the way logical Pauli states are initialized: to prepare logical $\lket{0}$ ($\lket{+}$), all of the physical qubits in the code patch are initialized in $\ket{0}$ ($\ket{+}$). 
As a result, stabilizers in the initialization basis are in their $+1$ eigenstates, but stabilizers in the opposite basis are $50/50$ random.
These random stabilizers cannot be reliably assigned in only $O(1)$ rounds. 
Therefore, to maintain fault tolerance against stabilizer measurement errors, we cannot rely on their information.
Analogously, an important assumption we make is that the magic $\ket{\lop{T}}$ states have been  prepared such that they have reliable stabilizers, which can be realized via a variety of approaches such as magic state cultivation~\cite{gidney2024magic} and distillation~\cite{bravyi2005universal}.

Crucially, when we back-propagate a reliable logical Pauli product through the circuit, by definition any logical Pauli initializations it terminates on must be in the same basis, which has $+1$ stabilizers.
Therefore, when decoding, we never use information from these 50/50 random stabilizers. 
Then, as the operator evolves through the transversal Clifford circuit, the stabilizer checks evolve identically, with noisy stabilizer measurements providing protection in the right basis against errors during transversal gates.
At the final transversal measurement of logical $\lop{Z}$ or $\lop{X}$, upon which all data qubits are measured in either the $X$ or $Z$ basis, respectively, the relevant stabilizers are inferred reliably from the data qubit measurements.

As a result, the decoding problem for the logical Pauli product bears key similarities to decoding a single logical qubit initialized in $\lket{0}$ and measured in the $Z$ basis.  
The $X$ stabilizers are not necessary to protect $\lop{Z}$ and, in fact, never need to be measured to accurately predict $\lket{0}$. 
Note that this approach crucially differs from existing matching-based decoders for transversal logical algorithms~\cite{beverland2021cost,sahay2024error,wan2024iterative}, which decode one logical qubit at a time and copy over error commitments.
Such strategies rely on the randomly-initialized stabilizers, which can lead to logical errors when only $O(1)$ SE rounds follow each operation (see Appendix~\ref{appendix:other_strategy}).

\begin{figure*}[t!]
\centering
\includegraphics{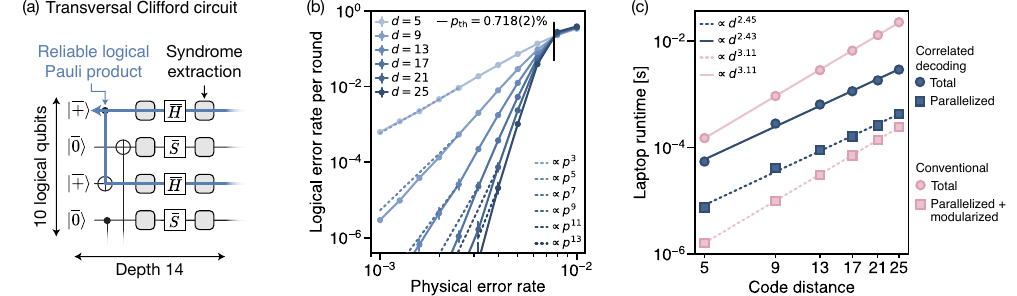}
\caption{Decoding transversal circuits with MWPM.
(a) We consider a depth-14 random logical Clifford circuit on 10 surface codes, consisting of alternating layers of transversal $\lcnot$s and fold-transversal $\lop{H}$ or $\lop{S}$ gates, each followed by a single SE round.
(b)~The logical error rate per gate layer as a function of physical error rate $p$ displays a threshold at $\pth = \pthclifforddeep$, and scales approximately as $(p/p_{\text{th}})^{\lfloor (d+1)/2\rfloor}$.
(c)~The total and parallelized run times for all reliable logical Pauli products at $p=0.1\%$ scale more favorably with code distance than conventional methods using lattice surgery and modular decoding, and the total run time is also reduced.
These run times are obtained using an Apple M2 Max laptop.
}
\label{fig:figure_3}
\end{figure*}

\subsection{Reduced decoding volume and software commitments \label{subsec:volume}}
Decoding only reliable logical Pauli products has the additional benefit of reducing the decoding volume (i.e., the size of the decoding problem). 
Each new measurement only requires decoding the part of the circuit that the single logical Pauli product traces through, as opposed to jointly re-decoding all logical qubits in the algorithm at each step (or their light-cones of depth $d$).
Furthermore, if many qubits are simultaneously measured, their corresponding logical Pauli product(s) can be decoded independently in parallel [see Fig.~\ref{fig:figure_2}(c)].
Interestingly, in this procedure the reliable logical Pauli products can have inconsistent physical error assignments.
However, their decoded values are still guaranteed to be correct (Theorem~\ref{thm:full_proof}), ensuring the fault tolerance of the algorithm as a whole. 

Because the reliable Pauli products are guaranteed to terminate on time boundaries with reliable stabilizers, this further opens up the possibility that the corrections can be committed upon decoding, such that the associated stabilizers do not need to be re-decoded in the future. 
This effectively minimizes the total decoding volume: once each region of space-time is passed through by the decoder, it will not be decoded again.
In Figure~\ref{fig:figure_4}, we apply this strategy to magic state distillation, empirically finding similar performance while ensuring that the factory does not need to be decoded again.

Although the results in Figure~\ref{fig:figure_4} indicate that this commitment strategy works in practice on a core algorithmic subroutine, there are various interesting routes to further analyze. 
In Appendix~\ref{appendix:software_commitments}, we prove that such a procedure is possible and the commitment boundary between reliable Pauli products during a $\lcnot$ experiences only local stochastic noise. 
However, more analysis is required to understand whether the resulting local stochastic noise is always matchable in circuits with \textit{time-like} loops, or undetectable error configurations within a reliable Pauli product that involve only stabilizer measurement errors (see Appendix~\ref{appendix:software_commitments} for additional discussion).
Furthermore, it would be interesting to understand whether the benefits of software commitments in reducing decoding volume can be combined with the parallelism of decoding the operators independently. 
Parallelization opportunities within a reliable Pauli product are also valuable to explore~\cite{fowler2015minimum,wu2023fusion}, as although many practical circuits have a linear structure~\cite{cuccaro2004new, babbush2018encoding,fowler2018low, haah2023quantum, zhou2025resource}, 
in the worst case the volume can grow exponentially~\cite{sahay2024error}.
It would be interesting to investigate how 
these various decoding considerations can be leveraged for circuit design in future algorithm compilations.

\section{Numerical Results}
\label{sec:numerics}
We now numerically simulate our decoding strategy, observing performance and run times comparable to those of a single-qubit memory.
Simulations are performed using the Stim package~\cite{gidney2021stim}, which enables circuit-level error sampling under a noise model with error probability $p$.
We use the PyMatching package for MWPM decoding~\cite{higgott2025sparse}.
Full details of the simulation setup, including the logical gates, noise model, and construction of the decoding hypergraphs, are provided in Appendix~\ref{appendix:numerics}.
The Stim circuits are available at Ref.~\cite{cain2025zenodo}.

In Figure~\ref{fig:figure_3}, we show the results of a benchmark involving a depth-14 transversal  Clifford circuit on 10 surface codes.
In odd gate layers, qubits undergo a random pairing of transversal $\lcnot$ gates; in even layers, half of the qubits randomly undergo $\lop{H}$ gates and the other half undergo fold-transversal $\lop{S}$ gates (see Fig.~\ref{fig:figure_3}(a)).
Each layer is followed by a single round of SE.
At the end of the circuit, we measure the stabilizers and the 10 reliable logical Pauli products using noiseless multi-qubit Pauli product measurements.
We choose the basis of reliable Pauli products which, when back-propagated through the circuit, terminate at a single logical Pauli initialization.
We decode these operators in parallel using independent matchable decoding graphs.
Note that this circuit has time-like loops (Sec.~\ref{subsec:volume}), so the strategies described in Appendix~\ref{appendix:software_commitments} must be applied to use the software commitment strategy.

The resulting logical error rate per gate layer~\footnote{The logical error rate per gate layer is defined in terms of the total logical error rate $P$ and circuit depth $D$ as \mbox{$P_\text{layer} = P_{\text{max}}\big(1 - (1 - P/P_{\text{max}})^{\frac{1}{D}}\big)
$}. 
Here, \mbox{$P_{\text{max}} = 1-2^{-10}$} is the total logical error rate of a fully mixed state of 10 logical qubits.} as a function of the physical error rate is plotted in Fig.~\ref{fig:figure_3}(b). 
From a fit to the data~\cite{wang2003confinement, watson2014logical}, we extract a threshold of \mbox{$\pth = \pthclifforddeep$}, which is similar to that of a single-qubit memory (\mbox{$\pth = \pthmemory$}; see Appendix~\ref{appendix:numerics}). 
Furthermore, the logical error rate per layer scales approximately as \mbox{$(p/p_{\text{th}})^{\lfloor (d+1)/2\rfloor}$}, consistent with achieving an effective code distance close to $d$.
We verify in Appendix~\ref{appendix:numerics} that these findings are robust to effects from the finite circuit depth by studying a deeper circuit sampled from the same distribution. 

Reducing the number of SE rounds by a factor of $O(d)$ not only decreases the resource cost of quantum operations; it also reduces the total amount of syndrome information in the circuit.
This raises the possibility that the total amount of computational resources required is, in fact, less than the conventional schemes based on lattice surgery, which require $O(d)$ SE rounds per gate for fault tolerance.
We analyze this possibility for the same transversal Clifford circuit in Fig.~\ref{fig:figure_3}(c), finding that the required total run time (throughput) is reduced compared to estimates for lattice-surgery-based computation with modular decoding~\cite{bombin2023modular}, and the parallelized run times (latencies) are comparable (see Appendix~\ref{appendix:surgery}).
Furthermore, because the run time is approximately proportional to the decoding volume (see Appendix~\ref{appendix:numerics}), our approach scales more favorably as a function of $d$ relative to the conventional setting.

To demonstrate our approach in a relevant algorithmic context, we evaluate its performance on a magic state distillation factory, shown in Fig.~\ref{fig:figure_4}(a)~\cite{bravyi2005universal,sales2024experimental}.
To efficiently numerically simulate the system, we substitute the distilled $\lket{T}$ states with $\lket{S}$ states.
A single round of SE is inserted between transversal gates, and the $\lket{S}$ states are prepared with an injected error probability $p$, followed by two rounds of noisy SE (see Appendix~\ref{appendix:numerics} for details).
The decoding is carried out in three stages, each following a layer of feed-forward gates.
After decoding each logical Pauli product, we commit the corresponding correction, reducing the decoding volume for subsequent stages.
As a result, no stabilizers in the factory need to be re-decoded when interpreting the final output qubit later in the computation. 
Figure~\ref{fig:figure_4}(b) shows that the commitment strategy achieves a threshold of $\pth = \pthdistillation$. 
As shown in Fig.~\ref{fig:figure_4}(c), the total decoding run time is substantially reduced compared to estimates for lattice-surgery-based computation with modular decoding (Appendix~\ref{appendix:surgery}).
As expected, the run time for the output qubit (stage 3) is over an order of magnitude smaller than the earlier stages, as its decoding volume is reduced from prior correction commitments.
For $d=25$ the entire factory is decoded in less than 100\,$\mu$s on an Apple M2 Max laptop. 

These numerical benchmarks show that in practice, correlated decoding can reduce to memory-like decoding. 
Therefore, the significant advances already made in tailored classical hardware for memories can readily be leveraged to decode transversal algorithms~\cite{liyanage2023scalable,barber2025real}.

\begin{figure}[t!]
    \centering
    \includegraphics{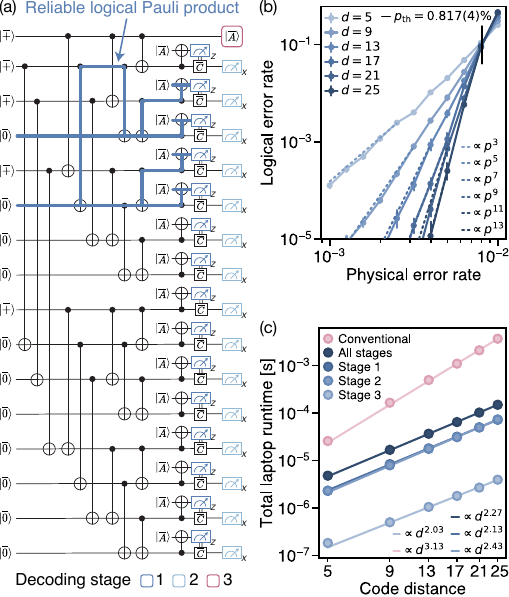}
    \caption{Magic state distillation.
    (a) The logical magic state distillation circuit converts 15 noisy $\lket{A}=\lket{T}$ states to a higher quality $\lket{T}$ state.
    To simulate the circuit efficiently, we replace the $\lket{T}$ states with $\lket{S}$ states and the $\lop{C}=\lop{S}$ feedforward gates with $\lop{Z}$ gates. 
    The factory is decoded in three stages, each following a layer of feed-forward gates.
    Upon decoding each reliable logical Pauli product (e.g., in blue), the corrections are committed, reducing the decoding volume of subsequent rounds. 
    (b)~We plot the measured fidelity of the distilled logical qubit, which is impacted by decoding all three stages correctly. 
    MWPM has a threshold of $\pth = \pthdistillation$ and scales approximately as \mbox{$(p/p_{\text{th}})^{\lfloor (d+1)/2\rfloor}$}.
    (c)~The total time to decode each stage on an Apple M2 Max laptop at a physical error rate of $p=0.1\%$ is substantially smaller than estimates for distillation with lattice surgery. 
    }
    \label{fig:figure_4}
\end{figure}

\section{Proof of fault tolerance\label{sec:proof}}

Paralleling the discussion above, we now prove that our approach is fault-tolerant and exponentially suppresses the logical error rate upon increasing code distance.
Our main result, Theorem~\ref{thm:full_proof}, is an alternative proof of transversal algorithmic fault tolerance for the surface code~\cite{zhou2024algorithmic} which now no longer relies on an inefficient most-likely-error decoder.

We consider the setting of universal computation with logical qubits encoded in unrotated surface codes. 
Clifford operations are implemented (fold-)transversally, each followed by a single SE round.
Concretely, Pauli and $\lcnot$ gates are implemented transversally, $\lop{H}$ is implemented with physical $H$ and a reflection about the diagonal, and $\lop{S}$ is implemented with $S$/$S^\dagger$ gates on the diagonal and $CZ$ gates between pairs of qubits reflected across the diagonal~\cite{kubica2015unfolding,moussa2016transversal,guernut2024fault,chen2024transversal} (see Appendix~\ref{appendix:decoding_hypergraph}, Fig.~\ref{fig:figure_7}). 
We assume the $\lket{T}$ magic states are fault-tolerantly initialized with reliable stabilizers, for example, as the output of magic state cultivation~\cite{gidney2024magic} or distillation~\cite{bravyi2005universal}.
Finally, we use a local stochastic noise model, in which the probability of $k$ elementary errors occurring decays as $p^k$, where $p$ is the probability of an individual error. 
The set of elementary errors are chosen to be single-qubit Pauli $X$ and $Z$ errors before each SE round and on the input $\lket{T}$ states, as well as flips of syndrome measurement results, similar to a phenomenological noise model~\cite{dennis2002topological, fowler2012surface}.

We first state our main result in this setting, then detail the proof below. 
\begin{theorem}[Exponential error suppression for universal quantum computation]
\label{thm:full_proof}
Consider a logical circuit implemented with surface codes of distance $d$ that comprises transversal Clifford gates and reliable magic state inputs (with $+1$ stabilizer measurement outcomes, up to local stochastic noise).
Then, there exists a threshold $p_0 > 0$, such that for local stochastic noise with constant physical error rate $p < p_0$, there is a decoding strategy based on MWPM with a logical error rate scaling with $d$ as $O(n_{\mathrm{loc}} m(p/p_0)^{d/4})$, where $n_{\mathrm{loc}}$ is the number of physical error locations and $m$ is the number of logical measurements.
\end{theorem}
\noindent 
We note that the factor of two reduction in the distance is due to a loose bound on error propagation of the fold-transversal $\lop{H}$ and $\lop{S}$ gates, and can likely be improved, as supported by the numerical evidence in Section~\ref{sec:numerics}.

We first show that each new logical measurement in the computation can be decoded reliably. 
To do so, we identify a basis of reliable Pauli products.
Suppose that $\nmeas$ measurements have occurred so far in a computation on $n$ logical qubits.
We associate a product of measurements with a  vector $\vec{v}\in \mathbb{Z}_2^{\nmeas}$ whose $i$th entry is one if and only if the product includes the $i$th measurement.
We then construct a matrix $M = \left[\frac{M^x}{M^z}\right]\in \mathbb{Z}_2^{2n}\times \mathbb{Z}_2^{\nmeas}$ describing how the measured operator back-propagates through the circuit: $M^x_{ij}$ ($M^z_{ij}$) is equal to one if and only if the back-propagation of measurement $j$ has support on the initialization of logical qubit $i$ in the $\lop{X}$ ($\lop{Z}$) basis.
Finally, we form a set $B$ characterizing which initialization bases are reliable.
We associate the $\lop{Z}$ (respectively, $\lop{X}$) initialization basis of logical qubit $i$ with a vector $\vec{e}_{z,i}\in\mathbb{Z}_2^{2n}$ ($\vec{e}_{x,i}$), which has only the $i$th ($(n+i)$th) entry equal to one. 
If the $i$th qubit is initialized in $\lket{T}$, $B$ includes both $\vec{e}_{z,i}$ and $\vec{e}_{x,i}$.
If the $i$th qubit is initialized in a Pauli state, $B$ only includes the initialization basis (e.g., $\vec{e}_{z,i} \in B$ if the $i$th qubit is initialized in $\lket{0}$).
\begin{definition}[Reliable logical Pauli product]
\label{def:reliable_product}
We call a logical Pauli product \textit{reliable} if its corresponding vector $\vec{v}$ satisfies 
\begin{equation}
    M\vec{v} \in \text{span}\,B. \label{eq:nonrandom_operators}
\end{equation}
\end{definition}
\noindent Note that the vectors from reliable logical Pauli products form a linear space, meaning any product of reliable logical Pauli products is also reliable. 
Therefore, as proven in Appendix~\ref{appendix:proof}, we can form a basis of measurement products, each of which either needs to be decoded or is 50/50 random and independent of other products. 
\begin{lemma}[Complete basis of measurements]
\label{lemma:complete_basis}
For a set of $\nmeas$ logical measurements, there exists a full-rank matrix $V\in \mathbb{Z}_2^{\nmeas\times \nmeas}$, where each column $\vec{v}_i$ of $V$ corresponds to a logical measurement product, such that either $\vec{v}_i$ is a reliable logical product, or the result of $\vec{v}_i$ is independent from other columns and always 50/50 random.
\end{lemma}
\noindent After interpreting the logical measurements in this basis, we can apply $V^{-1}$ to obtain the logical measurement results of each individual logical qubit.

We now show that reliable logical Pauli products are indeed ``reliable": they can be inferred with logical error rate exponentially suppressed with the code distance $d$.
To do so, we identify three key properties of the \textit{decoding subgraph} for each reliable logical Pauli product (Lemmas~\ref{lemma:reliable_init},~\ref{lemma:complete_subgraph}, and~\ref{lemma:matchable}).
This subgraph is constructed by back-propagating the measured Pauli product through the Clifford circuit, and including only checks involving stabilizer measurement results for the same logical qubits and basis as the instantaneous logical operator (see Appendix~\ref{appendix:decoding_hypergraph}).
Only the physical error sources which can flip these checks are included in the decoding subgraph.

\begin{lemma}[Reliable stabilizer initialization]
\label{lemma:reliable_init}
All initial stabilizers in the decoding subgraph of a reliable logical Pauli product are $+1$.
\end{lemma}
\begin{proof}
This follows from Definition~\ref{def:reliable_product}, as the back-propagated operator either terminates on $\lket{T}$ states (which have reliable $+1$ stabilizers by assumption) or logical Pauli products in the same basis (which have $+1$ stabilizers) (see also Sec.~\ref{subsec:fault_tolerance}).
\end{proof}

\begin{lemma}[Completeness of decoding subgraph]
\label{lemma:complete_subgraph}
Any elementary physical error that can affect the reliable logical Pauli product will be detected in the resulting decoding subgraph.
\end{lemma}
\begin{proof}
Consider any elementary Pauli error $e$, and denote its forward-propagation through the circuit as the Pauli operator $e'=U^\dagger eU$.
In order for the elementary error to affect the logical Pauli product $\lop{P}$, $e'$ and $\lop{P}$ must anti-commute $\{e',\lop{P}\}=0$, which in turn means that the error must anti-commute with the back-propagation of the logical Pauli product, $\{e,U\lop{P}U^\dagger\}=0$.
Since the stabilizer measurements are in the same basis as the logical Pauli product, the error $e$ must flip subsequent stabilizer measurement results and therefore be detected by the decoding subgraph.
\end{proof}

\begin{lemma}[Matchable decoding subgraph]
The decoding subgraph for any reliable logical Pauli product has edge degree at most two, and bounded vertex degree $O(1)$. \label{lemma:matchable}
\end{lemma}
\begin{proof}
For the surface code, data qubit $X$ or $Z$ errors trigger one or two syndromes at their endpoints, thereby affecting one or two checks.
As described in Sec.~\ref{subsec:matchable} and Appendix~\ref{appendix:decoding_hypergraph}, each syndrome measurement is involved in at most two checks.
Therefore, time-like edges also flip at most two checks, so the resulting edge degree is at most two.
Because all checks are formed from local products of measurement results, the number of elementary errors that can flip them will be bounded, leading to the bounded vertex degree.
\end{proof}
\noindent In Lemma~\ref{lemma:fewer_se} in Appendix~\ref{appendix:proof}, we show that similar conclusions also hold with less frequent SE, as long as the removed SE rounds do not form large connected clusters in the decoding subgraph.

These Lemmas are sufficient to establish that each reliable logical Pauli product can be predicted with exponentially low error probability. 
Because decoding subgraph is composed of simple edges of degree at most two, MWPM can be applied to efficiently identify the most likely error within the decoding subgraph.
This, along with the sparsity of the decoding subgraph, allows us to use standard cluster counting arguments~\cite{kovalev2013fault,gottesman2013fault} to bound the logical error rate of each reliable logical Pauli product (see Appendix~\ref{appendix:proof} for the full proof).

\begin{theorem}[Exponential error suppression for a single reliable Pauli product]
\label{thm:single_product}
Consider a logical circuit implemented with surface codes of distance $d$ that comprises transversal Clifford gates and reliable magic state inputs (with $+1$ stabilizer measurement outcomes, up to local stochastic noise).
Then, there exists a threshold $p_0 > 0$, such that for local stochastic noise with constant physical error rate $p < p_0$, any reliable logical Pauli product can be decoded with a logical error rate scaling with $d$ as $O(n^{\mathrm{sub}}_{\mathrm{loc}} (p/p_0)^{d/4})$ by applying MWPM to the corresponding subgraph, where $n^{\mathrm{sub}}_{\mathrm{loc}}$ is the number of physical error locations in the subgraph.
\end{theorem}

Because the logical error rate of each reliable logical Pauli product is exponentially small, and the 50/50 random operators are sampled via an unbiased coin flip, we use a union bound to show that the logical error rate of the algorithm is also exponentially small (Appendix~\ref{appendix:proof}, Lemma~\ref{lemma:union}). 
Note that the number of physical error locations in any decoding subgraph $n^{\mathrm{sub}}_{\mathrm{loc}}$ is bounded by the total number of physical error locations $n_{\mathrm{loc}}$.
This yields our final result, Theorem~\ref{thm:full_proof}, establishing that universal computation can be performed with $O(1)$ SE rounds per transversal logical operation using an efficient MWPM decoder.

\section{Conclusion and outlook}
In this work, we developed a procedure for decoding transversal logical algorithms, demonstrating that individual decoding of multi-logical Pauli products can preserve fault tolerance while greatly simplifying this computational task. 
For the surface code, the procedure results in a matchable decoding problem, substantially reducing its run time.
More generally, by tracking how logical operators propagate through a transversal circuit, this technique offers a path to promote a memory decoder to an algorithm decoder, while maintaining memory-like performance with $O(1)$ SE round per transversal logical operation.

These results can be extended in several directions. 
State-of-the-art QEC techniques developed for the memory setting can be applied to algorithms using this framework, including machine learning decoders~\cite{torlai2017neural, liu2019neural, varsamopoulos2020comparing, kuo2022comparison, bausch2023learning, cao2023qecgpt, maan2024machine, wang2023transformer, ninkovic2024decoding}, leakage detection and erasure decoding~\cite{wu2022erasure,kubica2023erasure,yu2024processing,chang2024surface,perrin2024quantum, baranes2025leveraging, gu2025fault}, parallelization methods~\cite{fowler2015minimum, wu2023fusion}, and high-rate quantum low-density parity check codes~\cite{breuckmann2021quantum,bravyi2024high,tillich2014quantum,xu2024constant}.
Furthermore, constructing matchable decoding problems from hypergraph problems can be viewed through the lens of generalized symmetries~\cite{kubica2023efficient, brown2022conservation}, offering additional perspectives and routes toward fast decoding in general quantum low-density parity check codes.
Finally, these techniques can be optimized for core algorithmic subroutines~\cite{cuccaro2004new,bravyi2005universal,babbush2018encoding, gidney2024magic} in several ways, including deciding whether software commitments and modular decoding are used, specifying the basis of reliable Pauli products to decode, and determining accurate methods for handling correlated errors.
Tailoring the decoding strategy to both the algorithmic structure and experimentally realistic error models are essential to realizing the optimal performance of fault-tolerant hardware.
\\

\textit{Note.---} During the preparation of this manuscript, we became aware of a similar and complementary work by the Delft group~\cite{serra-peralta2025decoding}. 
\\

\begin{acknowledgments}
\textit{Acknowledgements.---}
We thank B.~Brown, C.~Duckering, J.~Haah, R.~Ismail, M.~Kornjaca, S.~Puri, K.~Sahay, M.~Serra-Peralta, M.~Shaw, B.~Terhal, and S.~Wang for helpful discussions. 
We acknowledge financial support from IARPA and the Army Research Office, under the Entangled Logical Qubits program (Cooperative Agreement Number W911NF-23-2-0219), the DARPA ONISQ program (grant number W911NF2010021), the DARPA MeasQuIT program (grant number HR0011-24-9-0359), the Center for Ultracold Atoms (a NSF Physics Frontiers Center, PHY-1734011), the National Science Foundation (grant numbers PHY-2012023 and  CCF-2313084), the NSF EAGER program (grant number CHE-2037687), the Army Research Office MURI (grant number W911NF-20-1-0082), the Army Research Office (award number W911NF2320219 and grant number W911NF-19-1-0302), the Wellcome Leap Quantum for Bio program, and QuEra Computing.
D.B. acknowledges support from the NSF Graduate Research Fellowship Program (grant DGE1745303) and The Fannie and John Hertz Foundation.
\end{acknowledgments}

\newpage

\appendix

\newpage
\setcounter{page}{11}
\makeatletter
\renewcommand{\theequation}{S\arabic{equation}}
\renewcommand{\thefigure}{S\arabic{figure}}
\renewcommand{\thetable}{S\arabic{table}}
\renewcommand{\refname}{}

\section{Example of the decoding strategy\label{appendix:decoding_example}}

Here we give an explicit example of our decoding procedure. 
We study a circuit for small-angle synthesis~\cite{dawson2005solovay} using alternating $\lop{T}$ and $\lop{H}$ gates, as illustrated in Fig.~\ref{fig:figure_5}. 
The circuit and decoding procedure occur in the following steps.

First, a $\lop{T}$ gate is teleported onto qubit 1 using a magic state on qubit 2 [Fig.~\ref{fig:figure_5}(a)].
The measurement of $\lop{Z}_2$ anti-commutes with the $\lket{+}$ initialization of qubit 1 when back-propagated through the circuit (pink line).
Therefore, it can be assigned a $\pm 1$ value uniformly at random.

The next step of the circuit depend on which measurement assignment was chosen.
If the measurement was assigned as $-1$, a feed-forward $\lop{S}$ gate is applied; otherwise, no feed-forward occurs [Fig.~\ref{fig:figure_5}(b)]. 
Finally, a transversal $\lop{H}$ gate is performed, and another $\lop{T}$ gate is teleported onto qubit 1 using qubit 3. 
The procedure for interpreting the qubit 3 measurement depends on which feed-forward ``branch'' was chosen:
\begin{itemize}
    \item Fig.~\ref{fig:figure_5}(b), top: If the feed-forward $\lop{S}$ was applied, then $\lop{Z}_2 \lop{Z}_3$ is a reliable Pauli product, as its back-propagation through the circuit terminates on the $\lket{+}$ initialization of qubit 1 in the same basis and the qubit 2 and 3 magic states.
    $\lop{Z}_2 \lop{Z}_3$ can therefore be decoded reliably using only the stabilizer measurements along the operator back-propagation~(blue line).
    The $\lop{Z}_3$ measurement is then given by the product of the previous $\lop{Z}_2$ assignment and the decoded value of $\lop{Z}_2 \lop{Z}_3$.
    \item Fig.~\ref{fig:figure_5}(b), bottom: If no feed-forward gate was applied, then $\lop{Z}_3$ itself is a reliable Pauli product. 
    It can be decoded using stabilizer measurements along its back-propagation through the circuit (blue line).
    The $\lop{Z}_3$ measurement is then assigned from its decoded value.
\end{itemize}

In both cases, the reliable logical Pauli products are decoded fault-tolerantly using the procedure described in Section~\ref{subsec:decoding_strategy} of the main text.
Note that which logical Pauli products are reliable depends on which feed-forward branch was taken.
However, in any branch, new measurements are guaranteed to be either independently 50/50 random or part of a reliable Pauli product.
Therefore, they are always interpreted correctly according to the ideal joint logical measurement distribution.

\begin{figure}[t!]
    \centering
    \includegraphics{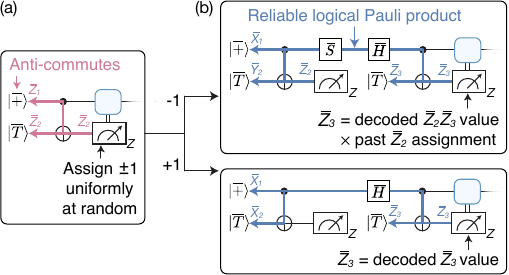}
    \caption{Example of the decoding strategy.
    (a)~The first measurement in this circuit is a $Z$-basis measurement on logical qubit 2.
    Because it anti-commutes with the $\lket{+}$ initialization when back-propagated through the circuit, it can be assigned a $\pm 1$ value uniformly at random.
    (b) If the first measurement is assigned as $-1$, a feed-forward $\lop{S}$ gate is applied (top).
    The measurement of qubit 3 is then assigned using the decoded value of the reliable logical Pauli product $\lop{Z}_2\lop{Z}_3$, and the previous assignment of $\lop{Z}_2$.
    Conversely, if the first measurement is assigned as $+1$, no feed-forward is applied and $\lop{Z}_3$ is a reliable Pauli product (bottom). 
    The measurement of qubit 3 is assigned from its decoded value.
    }
    \label{fig:figure_5}
\end{figure}

\section{Decoding hypergraph construction\label{appendix:decoding_hypergraph}}
Here we describe how to construct matchable decoding subgraphs for the unrotated surface code. 
Similar strategies can be applied to the rotated surface code and other CSS codes, as explored numerically in Fig.~\ref{fig:figure_4} of the main text.
To simplify the discussion, we assume that at each time step $t$, a SE round is performed on all logical qubits, and at most one transversal gate is inserted between rounds. 
We discuss generalizations to multiple gates per SE round in Lemma~\ref{lemma:fewer_se}.

We first define the decoding hypergraph $G=(V, E)$ for the full circuit.
The vertices $V$ are the checks, and the hyperedges $E$ are the error mechanisms, which connect the check(s) they flip.
We set a check for each stabilizer measurement by back-propagating its associated Pauli operator through the preceding transversal gate.
The check is then the product of the stabilizer measurement with the measured value(s) of its backwards-propagated operator at the previous time step.
In the absence of noise, the measurements should match, and all checks are deterministically $+1$.
An error which anti-commutes with an odd number of stabilizers in a check will flip its value to $-1$ and be detected.
Note that the checks for the first stabilizer measurements in the $+1$ initialization basis can form checks on their own, as they are already deterministically $+1$.
In contrast, the first measurement of the non-deterministic stabilizers cannot have an associated check.

Figure~\ref{fig:figure_6} explicitly shows the checks for a transversal $\lcnot$, $\lop{S}$, and $\lop{H}$ gate which, along with the transversal $\lop{X}$ and $\lop{Z}$ gates, generate the full logical Clifford group. 
In these examples, we include two rounds of SE before the transversal gate in order to easily visualize two sets of checks at adjacent time steps (here, comparing stabilizer measurements at times $t-1, t$ and $t, t+1$).
Fig.~\ref{fig:figure_6}(a) shows a transversal $\lcnot$ circuit with logical qubit 1 as the control and 2 as the target, and Fig.~\ref{fig:figure_6}(b) shows the corresponding decoding hypergraph.
Because all of the stabilizers in a given basis evolve identically, without loss of generality we focus on the same $Z$ stabilizer on both logical qubits.
The checks are then given by
\begin{align}
    \mathcal{Z}_1^{t-1, t} &= Z_1^{t-1} Z_1^{t} & \mathcal{Z}_1^{t, t+1} &= Z_1^{t} Z_1^{t+1} \nonumber \\
    \mathcal{Z}_2^{t-1, t} &= Z_2^{t-1} Z_2^{t} & \mathcal{Z}_2^{t, t+1} &= Z_1^{t} Z_2^{t} Z_2^{t+1},
\end{align}
where $\mathcal{Z}_i^{t-1, t}$ ($Z_i^{t}$) denotes the check ($Z$ stabilizer measurement) on logical qubit $i$ at time $t$. 
Similarly, Fig.~\ref{fig:figure_6}(c) and (d) show a logical $\lop{H}$ gate with relevant checks  
\begin{align}
    \mathcal{Z}_{y, x}^{t-1, t} &= Z_{y, x}^{t-1} Z_{y, x}^{t} & \mathcal{X}_{x, y}^{t, t+1} &= Z_{y, x}^{t} X_{x, y}^{t+1}\\
    \mathcal{X}_{y, x}^{t-1, t} &= X_{y, x}^{t-1} X_{y, x}^{t} &  \mathcal{Z}_{x, y}^{t, t+1} &= X_{y, x}^{t} Z_{x, y}^{t+1},
\end{align}
where $Z_{x, y}^{t}$ ($X_{x, y}^{t}$) represent a $Z$ ($X$) stabilizer at coordinates $(x, y)$ at time $t$, and the checks are defined analogously. 
Finally, Fig.~\ref{fig:figure_6}(e) and (f) show a logical $\lop{S}$ gate with relevant checks
\begin{align}
    \mathcal{X}_{x, y}^{t-1, t} &= X_{x, y}^{t-1} X_{x, y}^{t} & \mathcal{X}_{x, y}^{t, t+1} &= Z_{y, x}^{t} X_{x, y}^{t} X_{x, y}^{t+1}\\
    \mathcal{Z}_{y, x}^{t-1, t} &= Z_{y, x}^{t-1} Z_{y, x}^{t} & \mathcal{Z}_{y, x}^{t, t+1} &= Z_{y, x}^{t} Z_{y, x}^{t+1}.
\end{align}

\begin{figure}[t!]
    \centering
    \includegraphics{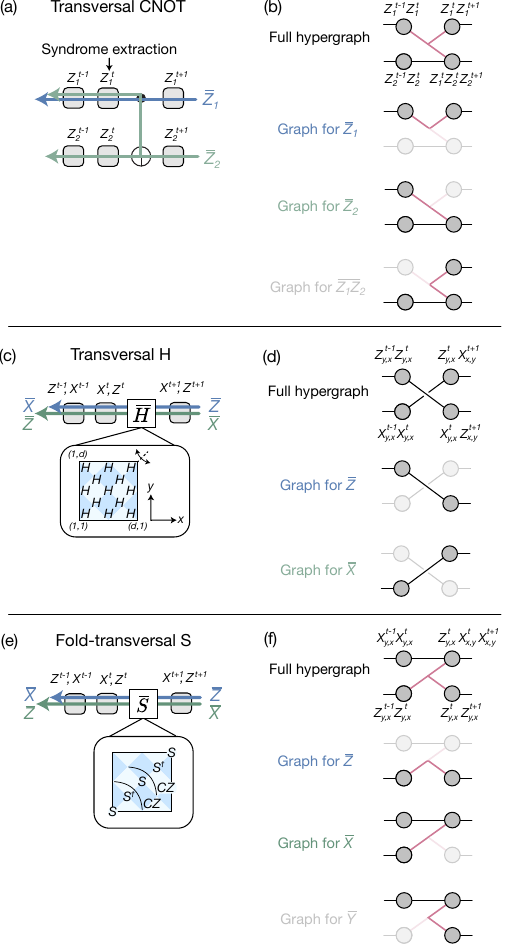}
    \caption{Decoding hypergraphs for transversal Clifford gates.
    We show the transversal $\lcnot$ (a, b), $\lop{H}$ (c, d), and $\lop{S}$ gates (e, f) and their corresponding decoding hypergraphs.
    When restricted to a single logical Pauli product of interest, the hyperedges are reduced to simple edges.
    }
    \label{fig:figure_6}
\end{figure}

Consider decoding a circuit with $m$ reliable Pauli products $\lop P_1, \dots, \lop P_m$.
From the full decoding hypergraph, we can construct a matchable decoding subgraph \mbox{$G_i = (V_i, E_i)$} for the $i$th reliable Pauli product.
To construct the decoding subgraph, we back-propagate $\lop P_i$ through the entire logical circuit (or for a depth equal to the code distance). 
We let $\lop P^t_i = \lop O_1^t \otimes \dots \otimes \lop O_n^t$ denote the logical observable over $n$ logical qubits at time $t$ during this back-propagation, where $\lop O_j^t \in \{\lop X_j, \lop Y_j, \lop Z_j, \lop I_j\}$. 
$V_i$ then contains checks for each $\lop O^t_j$ in the back-propagation.
Concretely, if $\lop O_j^t$ is $\lop X$, $\lop Y$, or $\lop Z$, then the checks  $\mathcal{X}_j^{t-1, t}$, $\mathcal{X}_j^{t-1, t}$ and $\mathcal{Z}_j^{t-1, t}$, or $\mathcal{Z}_j^{t-1, t}$ are included, respectively. 
$E_i$ then contains any errors that flip a check in the decoding subgraph (only their action on the subgraph checks is recorded).
This choice ensures that between these time steps, an error which anti-commutes with the logical observable will also be detected by the corresponding checks~(Lemma~\ref{lemma:complete_subgraph}).
One can verify using Fig.~\ref{fig:figure_7} that the resulting subgraph is always matchable.
Note that the subgraphs of \textit{all} logical Pauli products, even the unreliable ones, are matchable.

\section{Comparison with other strategies for matchable decoding\label{appendix:other_strategy}}

Here we compare our approach with previous strategies for decoding transversal gates with MWPM in Refs.~\cite{wan2024iterative,sahay2024error,beverland2021cost}.
%
These works decode one logical qubit at a time, then copy the error assignments between these logical qubits iteratively based on how transversal gates propagate physical $X$ and $Z$ errors.
For example, for the transversal $\lcnot$ circuit in Fig.~\ref{fig:figure_7}(a), one would first decode $Z$ stabilizers on the control qubit using MWPM, then commit the error assignments and update the corresponding checks on the target qubit.
For circuits involving multiple transversal $\lcnot$s, Ref.~\cite{wan2024iterative} proposes an iterative strategy in which this process is repeated until convergence.

However, for certain circuits, these approaches are not fault-tolerant with $O(1)$ SE rounds per logical operation. 
This originates from the fact that if the control qubit is initialized in $\lket{+}$, its $Z$ stabilizers cannot be reliably learned without $O(d)$ buffer SE rounds.
This can lead to high-weight erroneous assignments which are then copied onto other qubits.
Consider, for example, the circuit in Fig.~\ref{fig:figure_7}(a), which prepares a Greenberger–Horne–Zeilinger (GHZ) state among three logical qubits then measures $\lop{Z}_2\lop{Z}_3$, which should ideally be $+1$.
We first decode the $Z$ stabilizers ($X$ errors) of the top logical qubit, which is initialized in $\lket{+}$.
The $Z$ stabilizer values of the $\lket{+}$ state are unreliable with only a single SE round, resulting in unreliable error assignments which are then copied over to the remaining logical qubits, as illustrated in Fig.~\ref{fig:figure_7}(c).
A single data qubit error on one of the code patches then leads to the assignments of $\lop{Z}_2$ and $\lop{Z}_3$ differing with a probability that does not decay with the code distance.

\begin{figure}[t!]
    \centering
    \includegraphics{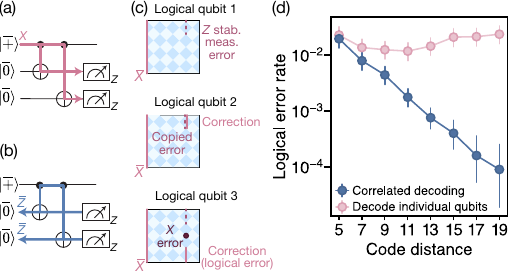}
    \caption{Comparison between decoding strategies for $O(1)$ SE round per gate.
    (a)~When preparing a GHZ state, $X$ errors on the top patch (pink) can propagate onto the $\lop Z$ measurements.
    (b)~We illustrate the back-propagation path of the reliable logical Pauli product $\lop{Z}_2\lop{Z}_3$. 
    (c)~There exists a failure mechanism with only two errors when logical qubits are decoded iteratively. 
    First, a $Z$ stabilizer measurement error on logical qubit 1 causes a string of $X$ errors.
    These errors are then copied onto logical qubits 2 and 3.
    Then, a single $X$ error on logical qubit 3 causes a logical error, as it violates  $\lop{Z}_2\lop{Z}_3$.
    (d)~Numerical simulations confirm that directly decoding the reliable logical Pauli product exponentially suppresses the logical error rate, whereas decoding individual logical qubits iteratively does not.
    }
    \label{fig:figure_7}
\end{figure}

In Figure~\ref{fig:figure_7}(d), we numerically simulate this circuit with rotated surface codes and a single SE round following transversal gates.
We apply circuit-level noise with error probability  $p=0.3\%$.
As expected, the logical error rate is suppressed exponentially in $d$ when decoding the reliable logical Pauli product shown in Fig.~\ref{fig:figure_7}(b) (blue), but is not if one decodes the logical qubits individually and transfers error assignments between them (pink). 

\section{Decoding strategy with software commitments \label{appendix:software_commitments}}
Here we describe in detail the decoding strategy in which the error assignments of the reliable Pauli products are committed in software.
The resulting procedure is given in Algorithm~\ref{alg:iterativedecoding}.
We show that the associated decoding problem is guaranteed to be matchable in the absence of \textit{time-like loops}, or cycles of stabilizer measurement errors which can be constructed from errors in a decoding subgraph.
In circuits with time-like loops, we find examples where the decoding problem is matchable, and other examples that are not directly matchable with our algorithmic procedure, but may be with additional modifications.
Because the committed errors can flip stabilizers on reliable Pauli products which have not yet been decoded, in the following section we prove that the resulting effective noise is local stochastic in an example setting of a transversal $\lcnot$.
Based on these findings and the magic state distillation simulations in Fig.~\ref{fig:figure_4} of
the main text, we conjecture that this procedure should have a threshold as long as the density of such commitment boundaries is sufficiently low. 

For concreteness of discussion, we focus on decoding $Z$ stabilizers in circuits with only transversal $\lcnot$ gates.
We consider a phenomenological noise model with physical $X$ errors on the data qubits directly before SE and $Z$ stabilizer measurement errors.
We also consider the toric code for simplicity due to its lack of boundaries.
General surface code circuits with all error types and $\lop{S}$ and $\lop{H}$ gates can be handled similarly.

Suppose we wish to decode the logical operators $\lop P_1, \dots, \lop P_m$, each associated with a decoding subgraph $G_i = (V_i, E_i)$ (Appendix~\ref{appendix:decoding_hypergraph}). 
Let \mbox{$E^G_i = \{e^G: e\in E_i\}$} denote the subset of hyperedges in the full decoding hypergraph $G=(V, E)$ which can flip checks in $G_i$. 
By Lemma~\ref{lemma:matchable}, each simple edge $e\in E_i$ is a subset of a hyperedge $e^G\in E^G_i$.
If the decoding procedure succeeds, it identifies the error in each $E_i$ up to a \textit{subgraph stabilizer}, defined as an operator that leaves no syndrome in $V_i$ and does not affect the decoded outcome of $\lop P_i$.

We begin by describing the procedure for circuits without time-like loops, e.g.,~as illustrated in Fig.~\ref{fig:figure_8}(a). 
First, decode the subgraph $G_1$ using MWPM. 
This allows us to fix the values of the hyperedges $E^G_1$ in subsequent decoding problems. 
Apply this correction in software by flipping the value of any check incident to an odd number of hyperedges in $E^G_1$ identified to have occurred. 
Finally, decode $G_2$ using the edges $e\in E_2$ such that $e^G\notin E^G_1$. 
At this point, the checks $V_1\cup V_2$ are  satisfied. 
Proceed in this manner until all subgraphs $G_i$ have been decoded. 
Because each hyperedge is only assigned once and then fixed for the rest of the decoding, the total decoding volume is at most the space-time volume of the circuit.

\begin{figure}[ht]
	\centering
    \includegraphics[width=\columnwidth]{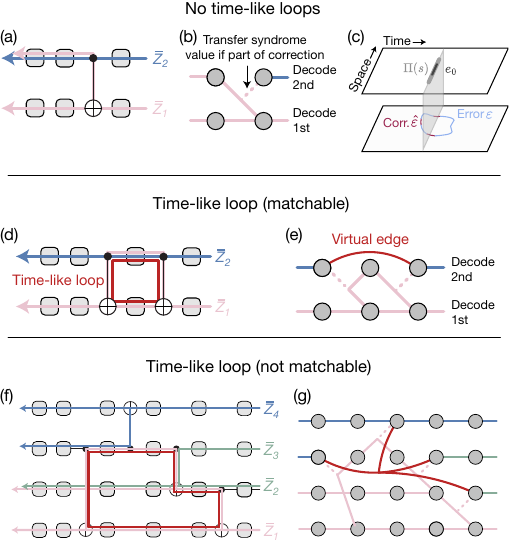}
    \caption{Software commitment strategy.
    (a)~An example of a logical circuit without time-like loops. 
    (b)~The corresponding decoding subgraphs after software commitments. 
    An error which can be transferred between subgraphs during commitment is shown with a dashed line.
    (c)~Data qubit $X$ errors and $Z$ stabilizer measurement errors in the bottom code block, jointly denoted by $\varepsilon$ (blue), and their MWPM correction $\hat\varepsilon$ (magenta).
    By committing these errors during a single transversal $\lcnot$ gate (gray slice), they result in the propagated data qubit $X$ errors $\Pi(s)$ (light gray) in the top code block.
    Note that $\varepsilon+\hat\varepsilon+\Pi(s)$ is an example of a connected cover for a subset $e_0$ of $\Pi(s)$ (dark gray).
    (d)~A logical circuit with a time-like loop (red), and the corresponding decoding hypergraph (e).
    After the pink edges are committed when decoding $\lop Z_1$, the decoding subgraph for $\lop Z_2$ (blue edges) should include a virtual edge (red) due to the time-like loop.
    (f)~A logical circuit, where decoding $\lop Z_1$ first results in a time-like loop (red).
    (g)~The corresponding decoding hypergraph.
    The time-like loop results in the weight-four virtual hyperedge in red, which restricts to an edge when decoding $\lop Z_2\lop Z_3$ (green) or $\lop Z_4$ (blue), but cannot be reduced when decoding $\lop Z_2\lop Z_3\lop Z_4$.
    \label{fig:figure_8}}
\end{figure}

We now discuss the performance of this procedure.
If the error rate is below the threshold in Theorem~\ref{thm:single_product}, $\lop P_1$ will be decoded correctly with high probability.
However, the residual subgraph stabilizer $s\subseteq E_1$ may leave a nontrivial syndrome in $G$ because of the transferred error hyperedges in $E^G_1$  [e.g., the dashed lines in the weight-three hyperedge in~Fig.~\ref{fig:figure_8}(b)]. 
Because $G_1$ does not contain any time-like loops, $s$ must include an even number of hyperedges at any given commitment boundary [Fig.~\ref{fig:figure_8}(c)]. 
Therefore, incorrect check values in subsequent decoding problems come in pairs, and they can be explained by an error $\Pi(s)$ connecting the checks.
As a result, future decoding rounds can see an elevated error rate at the commitment boundary. 
In the following section, we show that for a single CNOT decoded in two iterations, the propagated errors follow a local stochastic noise distribution. 
We conjecture that is true broadly for circuits which do not involve too many commitment boundaries, preventing the buildup of errors. 

Next, we consider the case with time-like loops in the decoding subgraphs [e.g.,~Figs.~\ref{fig:figure_8}(d) and \ref{fig:figure_8}(f)]. 
If the residual subgraph stabilizer $s$ is a time-like loop, it can result in a single incorrect check on a logical space-time block. 
Without knowing about the time-like loop, the flipped check may be matched to a time boundary that is far away in a subsequent decoding problem, leading to a large error and potentially an incorrectly-decoded observable. 
For example, in Fig.~\ref{fig:figure_8}(d), if an error from the red time-like loop flips the top left check, this excitation would be matched to the initialization boundary of the blue subgraph.

As a result, we will introduce ``virtual'' error mechanisms to provide future decoding problems with information about these time-like loops (Fig.~\ref{fig:figure_8}(e), red edge).
Concretely, we observe that any stabilizer $s$ of $G_1$ can be written as a sum $s = s_0 + s_1$, where $s_0$ does not contain any time-like loops and $s_1$ contains only time-like loops. 
The stabilizer $s_0$ is handled in the same way as in the first scenario by propagating an effective error to the later decoding iterations. 
Let $\{L_i\}$ be a basis of time-like loops. 
We introduce virtual error mechanisms $e_{i}$ corresponding to the elements $L_i$. 
More precisely, $e_{i} = \bigcup_{e\in L_i}e^G\setminus e$ is a hyperedge of $G$ consisting of the checks that would be triggered by $L_i$. 
Including $e_{i}$ as part of the matching when decoding subsequent observables indicates that we apply the error $L_i$ (which has no effect on the decoding of $\lop P_1$ since it is a subgraph stabilizer of $G_1$). 
The weights of the virtual hyperedges can be calculated based on the probabilities of the stabilizer loops occurring. 
In general, we expect the probability of $e_{i}$ to scale as $p^{|L_i|/2}$ to leading order, where $p$ is the probability of a measurement error, because at least half of the loop must have incurred an error when we decoded.
These probabilities can be updated based on the syndromes in $G_1$ seen by the decoder during the first decoding iteration.
We conjecture that these propagated errors can again be described with a local stochastic noise distribution, assuming the density of such errors is sufficiently low.
Algorithm~\ref{alg:iterativedecoding} summarizes the resulting decoding procedure.

If $\{L_i\}$ can be chosen so that the resulting virtual edges in future subgraphs all have weight two, then we may decode with MWPM as before. 
For example, this occurs in Fig.~\ref{fig:figure_8}(d), where each element contains at most two weight-three hyperedges.
However, for some circuits, such as the one illustrated in Fig.~\ref{fig:figure_8}(f), this is not possible, as some of the virtual hyperedges have weight at least three. 
In this case, in order to potentially maintain matchability, we would have to employ other methods such as decoding the observables in a different order, or choosing a different basis of reliable Pauli products.

\begin{algorithm}[H]
\caption{Decoding with software commitments}
\label{alg:iterativedecoding}
\textbf{Input:}\\
logical circuit $C$, measurements $\lop P_1, \dots, \lop P_m$ to be decoded, check values $d$, noisy measurement outcomes $\tilde a_1, \dots, \tilde a_m$\\
\textbf{Output:} \\decoded measurement values $a_1, \dots, a_m$
\begin{algorithmic}[1]
    \State $G = (V, E) \gets $ decoding graph from $C$
    \For{$i = 1, \dots m$}
        \State $G_i = (V_i, E_i) \gets $ subgraph for $\lop P_i$, calculated from $G$
        \State Adjust weights of edges in $E_i$, including virtual (hyper)edges, based on error propagation probabilities from decoding iterations $i' < i$ \Comment{Optional}
        \State $e_i \gets \operatorname{MWPM}(G_i, \left. d \right|_{V_i})$ \Comment{Only match using edges that have not already been fixed.}
        \State $a_i \gets $ decoded value of $\lop P_i$ from $\tilde a_i$ and $e_i$
        \State $\tilde e_i \gets \{e^G: e\in e_i\}$
        \State $d \gets d + \sigma(\tilde e_i)$ \Comment{$\sigma$ maps errors to the checks they flip}
        \State $\mathcal L \gets $ basis of physical time-like loops in $G_i$
        \State Add virtual hyperedges to $G$ for each element of $\mathcal L$
    \EndFor
    \State \Return $a_1, \dots, a_m$
\end{algorithmic}
\end{algorithm}

\subsection{Bounding error propagation in a single CNOT \label{app:error_prop}}
Now we prove a bound on the error propagation due to a single transversal $\lcnot$ gate in the toric code. 
We define the propagated error as follows. 
Let $\varepsilon$ be an error affecting the first logical space-time block and $\hat\varepsilon$ be the correction. 
Assuming $p$ is below the phenomenological noise threshold of the toric code, with high probability, $s=\varepsilon+\hat\varepsilon$ is a stabilizer of the first decoding subgraph $G_1$ consisting of homologically trivial cycles in space-time. 
Let $t=0$ be the time of the $\lcnot$, and $s'$ be the cycles that intersect the $t=-0.5$ surface. 
We define the propagated error $\Pi(s)$ by projecting (modulo 2) the qubit errors of $s'$ occurring at $t\ge 0$ (equivalently, $t\le -1$) onto the second logical space-time block. 
Note that there may be shorter equivalent errors that cause the same syndromes as the incorrectly committed hyperedges; nevertheless, we will prove that the error defined in this way is local stochastic.

\begin{theorem}\label{thm:local_stochastic}
Consider a logical circuit consisting of a single $\lcnot$ gate from the first qubit to the second. 
Suppose we decode $\lop Z_1$ followed by $\lop Z_2$, and both measurements are reliable. 
Under a phenomenological noise model of independent $X$ qubit errors and $Z$ measurement errors of sufficiently small strength $p$, the propagated errors from decoding $\lop Z_1$ can be described by a local stochastic noise model of strength $O(p)$.
\end{theorem}

\begin{proof}
For a given propagated error $e_0$, we consider all stabilizers $s$ that project to a superset of $e_0$. Without loss of generality, we may count only the stabilizers that are minimal, i.e., such that $e_0$ would not be contained in the projection if we removed any cycle from $s$.
To bound the probability of all such $s$, we will count clusters in a modified syndrome adjacency graph. Let $H_0$ by the syndrome adjacency graph of a single toric code time slice (with only qubit errors). 
Let $H$ be the syndrome adjacency graph of $G_1$ but with $H_0$ attached at the $t=-0.5$ time slice. 
We connect each qubit error of $H_0$ with the two adjacent measurement errors of $G_1$ at $t=-0.5$. Let $z$ be the degree of the graph $H$.
We consider $e_0$ and $\Pi(s)$ as subsets of vertices in $H_0$ and let $K = s \sqcup \Pi(s)$ be the subset of vertices in $H$. 
By definition of the error propagation, we have $|s| \ge \frac 2 3 |\Pi(s)|$. Let $\ell = |K|$ and $w = |e_0|$.

By minimality of $s$, the set $K$ is a \emph{connected cover} of $e_0$, meaning that it is a union of connected components in $H$, each of which contains an element of $e_0$. 
Otherwise, we could remove a component of $s$ that is disjoint from $e_0$. By Lemma 5 in Ref.~\cite{aliferis2007accuracy}, there are at most $e^{\ell - 1}z^{\ell - w}$ choices of connected covers $K$ of size $\ell$. 
These choices cannot all be decomposed as $s\sqcup \Pi(s)$, but all valid choices of $s$ that are minimal are included as one such connected cover.
For any $K$, we may recover $s$ as the restriction of $K$ to the errors in $G_1$. We use a union bound over all possible $s$ to control the probability of $e_0$.
\begin{align}
    \Pr(e_0) &\le \sum_{\substack{s: \Pi(s)\supseteq e_0\\s \text{ minimal}}}\Pr(s) \nonumber\\
    &\le \sum_{\substack{s: \Pi(s)\supseteq e_0\\s \text{ minimal}}} p^{|s|/2}2^{|s|} \nonumber\\
    &\le \sum_{K: \text{connected cover of } e_0}p^{|K|/3}2^{2|K|/3} \nonumber\\
    &\le \sum_{\ell=3w}^\infty e^{\ell-1}z^{\ell-w}p^{\ell/3}2^{2\ell/3} \nonumber\\
    &= \frac{2^{2w}e^{3w-1}z^{2w}p^w}{1 - 2^{2/3}ezp^{1/3}} \nonumber\\
    &= O\left((p/p_0)^w\right).
    \end{align}
The bound holds when $p < p_0/z$, where $p_0 = 1/(4e^3z^2)$. 
In the second inequality, we used the fact that the error weight must be at least half of the weight of the stabilizer $s$ by MWPM and that there are at most $2^{|s|}$ ways to choose such subsets of $s$. 
The third inequality can be seen by keeping only the terms where $K$ corresponds to a valid decomposition $s\sqcup \Pi(s)$ combined with the fact that $|s| \ge \frac 2 3 |\Pi(s)|$.
\end{proof}

\section{Details of the numerical simulations\label{appendix:numerics}}
Here we provide full details and additional benchmarking of the numerical simulations in the main text.
All circuit simulations and decoding hypergraphs were constructed using Stim~\cite{gidney2021stim}.
The circuit Stim files are available at Ref.~\cite{cain2025zenodo}.

\begin{figure*}[t!]
    \centering
    \includegraphics{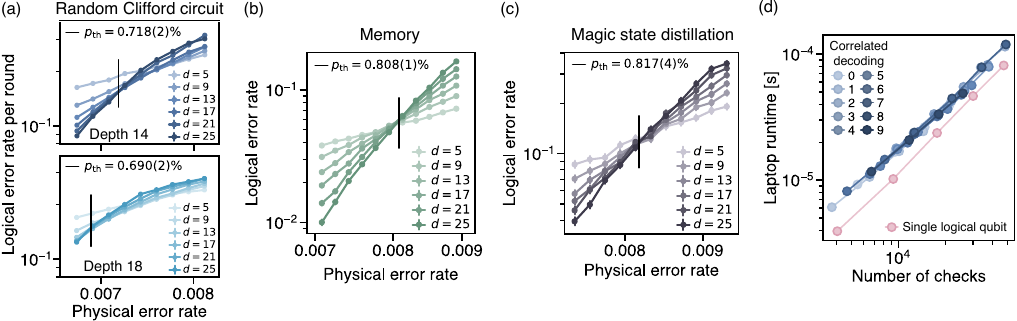}
    \caption{Benchmarking the performance and runtime.
    (a)~We show the MWPM thresholds for the Clifford circuit described in the main text (top), and for a depth-18 Clifford circuit randomly sampled from the same distribution.   
    (b)~A similar threshold is obtained for a single-qubit memory.
    (c)~We additionally show the threshold for the magic state distillation circuit described in the main text.
    (d)~The laptop run time for MWPM is proportional to the decoding volume (number of checks in the decoding subgraph) for both a single-qubit memory (pink) and correlated decoding of the 10 logical Pauli products of a random depth-10 Clifford circuit (blue). 
    We show data for code distances $d\in\{13, 17, 21, 25, 29\}$ for the single-qubit memory, and $d\in\{5, 9, 13, 17, 21, 25\}$ for the random Clifford circuit. 
    \label{fig:figure_9}
    }
\end{figure*}

Figure~\ref{fig:figure_3} in the main text explores a random transversal Clifford circuit on unrotated surface codes.
We first prepare half of the codes randomly in $\lket{0}$ or $\lket{+}$ by initializing the physical qubits in $\ket{0}$ or $\ket{+}$. 
We do not measure the surface code stabilizers as part of state preparation, as it is not necessary for fault tolerance (see Sec.~\ref{subsec:fault_tolerance}, main text).
Instead, we immediately begin applying the layers of transversal gates, which are implemented as illustrated in Fig.~\ref{fig:figure_7}. 
Each layer is followed by a single SE round, using the gate schedule in Ref.~\cite{acharya2023suppressing}. 
We decompose all logic gates and SE rounds into physical $Z$-basis measurement and reset gates, CNOT gates, or single-qubit gates ($H, S,$ and $S^\dagger$). 
We employ a circuit-level noise model with single-qubit depolarizing channels applied with probability $p$ during resets and measurements, and a two-qubit depolarizing channel with probability $p$ on each two-qubit gate.
After the final layer of transversal gates, we apply a noiseless round of SE in both the $X$ and $Z$ bases, then noiselessly measure all of the reliable logical Pauli products using multi-qubit Pauli product measurements. 
We set the checks using the procedure described in Appendix~\ref{appendix:decoding_hypergraph} and Fig.~\ref{fig:figure_7}.
The checks for the first SE round are set from the deterministic $+1$ stabilizers at initialization (or products of stabilizers, if the $+1$ stabilizers are copied between logical qubits during the first transversal $\lcnot$ gate layer).

Fig.~\ref{fig:figure_9}(a) shows the threshold of the depth-14 random Clifford circuit described in the main text (top) and a depth-18 circuit sampled from the same distribution (bottom).
To benchmark the decoding strategy, Fig.~\ref{fig:figure_9}(b) shows the threshold of a single-qubit memory with $d$ SE rounds in both bases and a layer of idling depolarizing noise between SE rounds.
By fitting distances $\geq 13$ to a universal scaling formula~\cite{wang2003confinement, watson2014logical}, we extract thresholds of $\pth = \pthclifforddeep$ and $\pth = \pthclifforddeeper$ for the depth-14 and depth-18 Clifford circuits, respectively, and $\pth = \pthmemory$ for the single-qubit memory.
The similarity of the random Clifford circuit thresholds suggests they are not sensitive to effects from the finite circuit depth, and are only modestly reduced compared to that of the memory circuit, consistent with the findings in Ref.~\cite{cain2024correlated}.

In Fig.~\ref{fig:figure_9}(d), we benchmark the decoding run time for a randomly sampled depth-10 Clifford circuit drawn from the same distribution as above.
We observe that the run time for each of the 10 reliable Pauli products scales approximately linearly with their decoding volume for different code distances, identical to the scaling of a single logical qubit.
Therefore, the run times for decoding transversal algorithms and a single logical qubit are comparable when normalized by the decoding volume.

The simulations for Fig.~\ref{fig:figure_4} in the main text are on rotated surface codes using the same gate set, stabilizer measurement circuit, and error model as Fig.~\ref{fig:figure_3}.
We prepare the distilled $\lket{S}$ patches with an injected error probability $p$ and a perfect round of SE, followed by two noisy rounds of SE to reach noise equilibrium. 
We decode the factory in three stages, and we apply feed-forward $\lop{Z}$ gates based on the decoded results.
To probe the fidelity of the output $\lket{S}$ state, we teleport a noiseless $\lop{S}$ gate then measure the output qubit in the $X$ basis. 
Furthermore, the decoding hypergraph construction for the circuit is modified so that errors can be committed in software.
We first construct the full decoding hypergraph, then remove checks not involved in the decoding subgraph of interest. 
This enables the action of each error on checks from different subgraphs to be recorded for software commitments.
To decode, we combine any error mechanisms with the same syndrome pattern when restricted to the decoding subgraph.
When committing these indistinguishable errors, the representative which flips the fewest checks on other reliable Pauli products is committed.
Fig.~\ref{fig:figure_9}(c) shows the resulting threshold for magic state distillation, extracted using the same methodology as the previous circuits.

\section{Proof details \label{appendix:proof}}
In this appendix, we provide the detailed proofs of some of the lemmas from the main text.

\begin{proof}[Proof of Lemma~\ref{lemma:complete_basis}]
We show this lemma by inductively constructing the basis $V$.
We will show that any column $\vec{v}_i$ that is not a reliable logical product must anti-commute with some logical Pauli stabilizer $s_i$, and that for different such columns, the set of anti-commuting logical Pauli stabilizers $\{s_i\}$ is linearly independent.
If this is true, then we can find a logical Pauli stabilizer that only flips the $i$th unreliable logical product: since this should not change the measurement distribution, we can conclude that its result is independent from other columns and must always be 50/50 random.

For the base case, consider the first logical Pauli measurement $\lop P_1$.
By Definition~\ref{def:reliable_product}, if $\lop P_1$ is not reliable, then it must anti-commute with some logical Pauli stabilizer, satisfying the condition above with $V=(1)$.

For the inductive case, suppose we have a full-rank matrix of basis vectors for the first $\nmeas$ measurements $V_\nmeas$ satisfying the conditions.
With the $(\nmeas+1)$th Pauli measurement $\lop P_{\nmeas+1}$, we update the basis as follows:
\begin{itemize}
    \item Extend the first $\nmeas$ columns of $V_\nmeas$ to length $\nmeas+1$, with 0 in the new row.
    \item If there exists a subset $S$ of earlier measurements (possibly empty), such that $\lop P_{\nmeas+1} \times \prod_{s\in S} \lop P_{s}$ is a reliable logical Pauli product, then include this product in the basis by setting the last column $\vec{v}_{\nmeas+1}$ to be 1 at rows in $S$ and the $(\nmeas+1)$th row, and 0 otherwise.
    \item Otherwise, include $\lop P_{\nmeas+1}$ in the basis by setting the last column $\vec{v}_{\nmeas+1}$ to be 1 at the $(\nmeas+1)$th row, and 0 otherwise.
\end{itemize}

The matrix $V_{\nmeas+1}$ constructed recursively this way will be upper triangular, with all diagonal elements being 1, so it is clearly full rank.
Now we show that it also satisfies the condition on unreliable Pauli products.
Suppose this is not the case, then there exists a set of logical Pauli stabilizers $\{s_i\}$, each anti-commuting with a different unreliable logical Pauli product $\vec{v_i}$, that is linearly dependent; hence, a subset of the $\{s_i\}$ product to the identity.
This subset must involve the newly added measurement, since otherwise it violates the induction hypothesis.
However, taking the product of the corresponding logical measurements results in a logical Pauli product involving $\lop P_{\nmeas+1}$ that commutes with all logical Pauli stabilizers and is therefore reliable, contradicting the construction of $\vec{v}_{\nmeas+1}$.
Therefore, the condition on unreliable Pauli products is satisfied, completing our proof.
\end{proof}

\begin{proof}[Proof of Theorem~\ref{thm:single_product}]
We follow the procedure in Ref.~\cite{gottesman2013fault} to bound the logical error rate.

Consider the syndrome adjacency graph, for which vertices $e_i$ are error events included in the decoding subgraph, and there is an edge $(e_i,e_k)$ between vertices if they both trigger the same check.
By Lemma~\ref{lemma:matchable}, each error is involved in at most two checks, each of which has a bounded vertex degree, so the vertices in the syndrome adjacency graph also have degree bounded by some constant $z$.
Errors and inferred corrections form undetectable clusters on the syndrome adjacency graph.
We will analyze properties of these connected clusters to bound the logical error rate.

We can lower bound the number of vertices in a connected cluster required for a logical error to occur.
For a fault cluster $f$ and logical Pauli product $\lop{P}$, we can propagate both of them back through the circuit until they are only supported on state initialization, resulting in operators $\tilde{f}$ and $\tilde{P}$.
Since we are propagating both faults and logical Pauli products through the same circuit, $f$ will lead to a logical error if and only if $\{\tilde{f},\tilde{P}\}=0$.
By Lemma~\ref{lemma:reliable_init}, all initial stabilizers in the same basis as $\tilde{P}$ have known initial values.
Any error cluster that can cause a logical error on $\tilde{P}$ must be in the opposite basis, and by the distance of the code, we have a lower bound on the weight $|\tilde{f}|\geq d$.
By the transversal structure of the circuit, in which a given error can only spread to at most two other spatial locations, this implies that $|f|\geq d/2$.

By Lemma 5 in ~\cite{aliferis2007accuracy} and Lemma 2 in~\cite{gottesman2013fault}, the number of clusters with weight $w$ that contain any given error is upper bounded by $(ze)^{w-1}$.
The number of physical error locations in the decoding subgraph is $n^{\mathrm{sub}}_{\mathrm{loc}}$.
Since the MWPM decoder is able to identify the minimum-weight error, the weight of the correction is upper bounded by the number of physical errors in each cluster.
Therefore, at least $w/2$ physical errors must have occurred, with a probability upper bounded by $p^{w/2}$.
For a cluster of size $w$, there are at most $2^w$ ways to choose subsets that correspond to the physical error configuration.
This allows us to bound the logical error rate $\plogical$ on the Pauli product as
\begin{align}
\plogical &\leq O(n^{\mathrm{sub}}_{\mathrm{loc}}) \sum_{w\geq d/2}(ze)^{w-1} 2^w p^{w/2} \nonumber\\&= O\qty(\frac{n^{\mathrm{sub}}_{\mathrm{loc}}}{ze} \frac{\left( 2ze \sqrt{ p } \right)^{d/2}}{1-2ze\sqrt{ p }})\nonumber\\
&=O\qty(n^{\mathrm{sub}}_{\mathrm{loc}}\qty(\frac{p}{p_0})^{d/4}),
\end{align}
when $p < p_0$, where $p_0=1/(2ze)^2$, as desired.
\end{proof}

\begin{lemma}[Reduced number of syndrome rounds]
Consider a surface code transversal Clifford circuit with one SE round per logical operation, denoted $\mathcal{C}$, and the decoding subgraph for any reliable logical Pauli product $\lop P$.
In this subgraph, two stabilizer measurements are called adjacent if they are involved in the same check.

Suppose we remove some of the stabilizer measurements to form the circuit $\mathcal{C}'$, such that the maximal cluster of adjacent measurements removed involves $r$ measurements.
Then for decoding $\lop P$ in the circuit $\mathcal{C}'$, we can construct a decoding subgraph with edge degree at most two, and vertex degree at most $5r+7$.
\label{lemma:fewer_se}
\end{lemma}

\begin{proof}
If we consistently use the check convention in Appendix~\ref{appendix:decoding_hypergraph}, then for the circuit $\mathcal{C}$, each check (vertex in the decoding hypergraph) has edge degree at most seven in our simplified error model from Sec.~\ref{sec:proof}: at most four error edges from data qubit errors happening before the syndrome measurement, and at most three error edges from measurement errors of the stabilizer measurements that constitute the check.
For the circuit $\mathcal{C}'$ with fewer SE rounds, we can start from the circuit $\mathcal{C}$ with a single SE round per gate and merge checks.
As shown in Lemma~\ref{lemma:matchable}, within the decoding subgraph of any reliable logical Pauli product $\lop P$ for a circuit $\mathcal{C}$, each measurement is only involved in two checks.

By replacing two checks triggered by a measurement error with a single check formed by their product, and by connecting any edges that were linked to either old vertex to the new vertex, we can remove a measurement from the decoding problem, while constructing a decoding graph that maintains the property of only having simple edges in the time direction.
Therefore, the subgraph remains matchable regardless of the frequency of SE.
The degree of the new vertex is upper bounded by the sum of the vertex degrees of the old vertices, minus the shared edge between the old vertices, which contributed twice to the original degree.
Therefore, each merge increases the vertex degree by at most five.
Since the maximal cluster of adjacent measurements removed is size $r$, the maximal vertex degree is at most $5r+7$.
\end{proof}

\begin{lemma}[Union bound on logical errors in transformed basis]
\label{lemma:union}
Consider a probability distribution $Q(\vec{x})$ over $\nmeas$ binary random variables $\vec{x}=(x_1,x_2,...,x_\nmeas)$, and a full rank matrix $V\in \mathbb{F}_2^{\nmeas\times\nmeas}$ representing a basis transformation of the random variables.
Suppose an error (possibly correlated) of probability at most $p$ is applied to each element of $V\vec{x}$, resulting in a new distribution $Q'(\vec{x})$.
Then the total variation distance between the original and erroneous distribution is upper bounded by $||Q(\vec{x})-Q'(\vec{x})||_{TV}\leq \nmeas p$.
\end{lemma}

\begin{proof}
Denote the vector $\vec{y}=V\vec{x}$, and the vector with noise applied $\vec{y}'=\vec{y}\oplus \vec{e}$.
Each element of $\vec{e}$ has probability at most $p$.
Applying the union bound, we have that $P(\vec{e}\neq 0)\leq mp$.
Since the matrix $V$ is full rank, we can invert the vector $\vec{y}'$ to obtain a vector $\vec{x}'$ in the original basis.
With probability $1-\nmeas p$, no error vector was applied, so we recover the same vector $\vec{x}'=\vec{x}$, implying that deviations in the distribution can only occur in the remaining $\nmeas p$ probability.
This implies the upper bound $||Q(\vec{x})-Q'(\vec{x})||_{TV}\leq \nmeas p$.
\end{proof}

\section{Throughput and latency estimates for alternative schemes}
\label{appendix:surgery}
Here we estimate the run times of modular decoding with lattice surgery on the circuits numerically explored in the main text. 
To determine the run times of the magic state distillation circuit in Fig.~\ref{fig:figure_4}, we leverage existing analyses of lattice-surgery-based magic state distillation, specifically Sec.~VII.E of Ref.~\cite{bombin2023modular}.
This approach consists of logical blocks (vertices) connected via lattice surgery (edges), and enables decoding arbitrarily-large computations in only two stages, where edges are first decoded with a sufficiently large buffer region and committed, and then vertices are decoded.
Following Ref.~\cite{bombin2023modular}, we choose the edges to have width $d$, so that different edges are sufficiently separated and can be committed simultaneously.
Using a buffer region of width $d/2$, the edges can be decoded in parallel by considering a decoding problem of size $d^3$.
After committing to the center of the edges, the vertices can then be decoded, where a vertex with $s$ edges connected will require a decoding problem of size at least $sd^3/2$.

With these considerations, we estimate the total amount of decoding work (total computational run time across different parallel cores) and decoding latency (also known as depth or span in parallel computing) for the lattice-surgery-based magic state distillation procedure in Fig.~13 of Ref.~\cite{bombin2023modular}.
We denote the time required to decode a depth $d$ logical memory experiment of volume $d^3$ as $T_0$, which sets the units for our run time estimation.
We assume that the decoding time scales linearly with the problem size, which has been found to be a good approximation in the regime of low physical error rates~\cite{higgott2025sparse,wu2023fusion}.
The maximal vertex degree in a magic state factory is 7, resulting in a decoding latency of $T_0$ for the first stage, and $7T_0/2$ for the second stage, for a combined latency of $9T_0/2$.
There are 46 edges, each of which will be covered twice (once during edge decoding, once during vertex decoding), so the total amount of decoding work across all parallel cores will be $92T_0$.
Here, we have neglected the volume of the vertices themselves, so our estimations represent a lower bound for this particular distillation factory construction.
We note that alternative choices of factory construction, edge width and buffer size~\cite{fowler2018low,gidney2019efficient} may lead to different trade-offs in decoding work and depth, which can be further analyzed in future work.
However, our qualitative conclusions should apply regardless of the details of the factory construction.

We can perform a similar analysis for the random transversal Clifford circuit in Fig.~\ref{fig:figure_3}.
Here, we adopt the strategy described in Refs.~\cite{sahay2024error,wan2024iterative,beverland2021cost}, in which transversal gates are separated by $d$ SE rounds, and for each transversal gate, one logical qubit is decoded before the other.
A modular decoding approach can then be applied, in which we first decode the $d$ SE rounds between each transversal gate and commit the correction in the center.
We can then decode each of the individual transversal gates in an ordered fashion.
For the random Clifford circuit in Fig.~\ref{fig:figure_3}(a), consisting of 10 surface code logical qubits and depth 14, the latency is $3T_0$, where the first stage of decoding the $d$ rounds requires time $T_0$ and the second stage of ordered decoding requires time $2T_0$.
The total amount of work is $280T_0$, which is twice the total decoding volume of SE, since we solve the decoding problem twice, one at each stage.

\bibliography{main.bib}

\end{document}